\documentclass[11pt,centertags,reqno,twoside]{amsart}
\usepackage{amsmath,latexsym, graphicx}
\usepackage[psamsfonts]{amssymb}
\usepackage{mathptmx}
\usepackage[bookmarks]{hyperref}
\usepackage[mathcal]{euscript}
\usepackage{xcolor}
\usepackage[normalem]{ulem}
\usepackage{amsfonts}
\usepackage{amssymb}
\usepackage{amsthm}
\usepackage{bm}
\usepackage{enumitem}
\usepackage[english]{babel}
\usepackage[bookmarks]{hyperref}

\numberwithin{equation}{section}

\renewcommand{\epsilon}{\varepsilon}

\newcommand{\be}{\begin{equation}}
\newcommand{\ee}{\end{equation}}

{\bf}{\it}
\newtheorem{theorem}{Theorem}[section]
\newtheorem{lemma}[theorem]{Lemma}
\newtheorem{corollary}[theorem]{Corollary}

\newtheorem{definition}[theorem]{Definition}

\newtheorem{remark}[theorem]{Remark}

\date{\today}
\begin{document}

\title[Explicit Bound States and Threshold Resonances for Two Identical Fermions...]{Explicit Bound States and Threshold Resonances for Two Identical Fermions on a One-Dimensional Lattice}

\author{S.S.Ulashov, Sh.I.Khamidov}

\address[S.S. Ulashov]{Samarkand State University, 140104, Samarkand, Uzbekistan}
\email{sobirulashov19@gmail.com}

\address[Sh.I. Khamidov]{V.I. Romanovskiy Institute of Mathematics, UzAS,Tashkent, Uzbekistan}
\email{shoh.hamidov2021@gmail.com}

\begin{abstract}
We study a two-particle lattice Schr\"odinger operator describing two
identical fermions on the one-dimensional lattice \(\mathbb{Z}\) with
a nearest-neighbor interaction of strength $\lambda\in \mathbb{R}$. The Hamiltonian acts in the
antisymmetric subspace
\(\ell^{2,\mathrm{a}}(\mathbb{Z}^{2})\). Using the direct-integral
decomposition with respect to the total quasi-momentum
\(k\in\mathbb{T}:=(-\pi,\pi]\), we reduce the spectral problem to a
family of fiber operators \(H_{\lambda}(k)\) acting in the odd
relative-coordinate space
\(\ell^{2,\mathrm{odd}}(\mathbb{Z})\). For every
\(k\in(-\pi,\pi)\), the fiber Hamiltonian is represented as a
one-dimensional discrete Schr\"odinger operator, equivalently, as a
Jacobi operator with a finite-range boundary perturbation,  and its
essential spectrum is
\(
\sigma_{\mathrm{ess}}\bigl(H_{\lambda}(k)\bigr)
=
\bigl[
4-4\cos\frac{k}{2},
\,4+4\cos\frac{k}{2}
\bigr].
\)\\
We give a complete description of the spectral transition at both
edges of the essential spectrum. For each $k\in(-\pi,\pi),$ the operator
\(H_{\lambda}(k)\) has a unique simple eigenvalue outside the essential
spectrum if and only if
\(|\lambda|>2\cos\frac{k}{2}.\)
This eigenvalue is given explicitly by
\(
E(k,\lambda)
=
4+\lambda
+
\frac{4\cos^{2}\frac{k}{2}}{\lambda},
\)
and the corresponding eigenfunction is obtained in explicit form in the
relative coordinate.
For
\(
|\lambda|<2\cos\frac{k}{2},
\)
there is no eigenvalue outside the essential spectrum.\\
The critical threshold transition is also analyzed for each $k\in(-\pi,\pi).$ At the critical coupling
\(|\lambda|=2\cos\frac{k}{2},\)
the discrete eigenvalue merges with the corresponding edge of the essential spectrum and gives rise to a threshold resonance. The corresponding threshold equation admits a nonzero bounded odd solution belonging to \(
\ell^{\infty,\mathrm{odd}}(\mathbb{Z})
\setminus
\ell^{2,\mathrm{odd}}(\mathbb{Z}),
\) showing that the spectral threshold is a resonance rather than an eigenvalue.
\\
We further establish the strong-coupling and threshold asymptotics of
the discrete spectrum. In particular,  in the strong-coupling regime the  eigenfunction
becomes localized at the interaction sites, whereas at the critical coupling it
converges pointwise to the corresponding threshold resonant solution.
\\
The exceptional fiber
\(k=\pi\), for which the relative hopping vanishes and the essential
spectrum collapses to the point \(\{4\}\), is treated separately.
\end{abstract}

\maketitle



\section{Introduction}

Lattice Schr\"odinger operators form an important class of models in
mathematical physics. They describe quantum particles moving in periodic
media and arise naturally in the study of few-particle systems in optical
lattices and solid-state models. Their spectral properties depend on the
interaction between the particles, the geometry and dimension of the
lattice, and the total quasi-momentum. Even in the two-particle case, this
interplay may produce isolated eigenvalues, virtual levels, and resonance
phenomena at the edges of the essential spectrum.

In the continuous setting, threshold phenomena for few-particle
Schr\"odinger operators are closely related to scattering theory, virtual
levels, and the singular behavior of the resolvent near the spectral
thresholds; see, for example, \cite{Yaf74,FMerkuriev:1993}. Lattice systems
exhibit analogous effects, but the boundedness of the lattice dispersion
relation produces both lower and upper spectral edges. Consequently,
attractive interactions may generate bound states below the continuous
band, whereas repulsive interactions may produce isolated eigenvalues
above it.

The few-body problem on lattices was systematically studied by Mattis
\cite{Mattis:1986}. The physical relevance of these models has also been
confirmed by investigations of ultracold atoms in optical lattices. In
particular, repulsively bound atom pairs were observed experimentally in
\cite{Winkler:2006}, and two-particle states in Hubbard-type lattice models
were studied in \cite{ValientePetrosyan:2008}.

Two-particle lattice Schr\"odinger operators have been investigated
extensively with respect to bound states, virtual levels, threshold effects,
and the dependence of the discrete spectrum on the total quasi-momentum;
see
\cite{ALMM:2006,LakaevUlashov2012,BachLakaevPedra:2017,
LakaevBachPedra:2020,KholmatovLakaevAlmuratov:2020,KhKhM:2024,
MuminovKhurramov2016,MuminovKhurramovBozorov2022}.
Models with finite-range interactions, including on-site,
nearest-neighbor, and next-nearest-neighbor interactions, were considered
in
\cite{LakaevKholmatovKhamidov2021,UlashovKhamidovLakaev2024,
LakaevKhamidovAkhmadova2024,AkhmadovaAzizova2025,
LakaevLatipovaAkhmadova2025,LakaevAbdukhakimovKhasanov2026}.
These works mainly use the momentum representation, finite-rank
perturbation theory, Fredholm determinants, and partitions of the
interaction-parameter space.

The fermionic two-particle lattice Hamiltonian with a nearest-neighbor
interaction was studied in \cite{LKh2009}, where the number and location of
its discrete eigenvalues were described in terms of the interaction
strength, the total quasi-momentum, and the lattice dimension. A related
two-dimensional model with first- and second-nearest-neighbor interactions
was considered in \cite{LakaevMotovilovAbdukhakimov:2023}, where the
interaction-parameter plane was partitioned according to the number of
eigenvalues below and above the essential spectrum.

The purpose of the present paper is to give an explicit description of the
threshold spectral transition for two identical fermions on the
one-dimensional lattice with a nearest-neighbor interaction. Our approach
is based on the relative-coordinate representation. After the
direct-integral decomposition with respect to the total quasi-momentum
\(
k\in\mathbb{T}:=(-\pi,\pi],
\) the antisymmetry of the two-particle wave function becomes the oddness
condition in the relative coordinate. The corresponding fiber Hamiltonian
acts in \(\ell^{2,\mathrm{odd}}(\mathbb{Z})\) as
\[
\bigl(H_{\lambda}(k)u\bigr)(r)
=
4u(r)
-
2\cos\tfrac{k}{2}
\bigl[u(r+1)+u(r-1)\bigr]
+
\lambda v(r)u(r),
\qquad r\in\mathbb{Z},
\]
where
\[
v(r)
=
\begin{cases}
1, & |r|=1,\\
0, & |r|\neq1.
\end{cases}
\]
Thus, every nondegenerate fiber is reduced to a one-dimensional discrete
Schr\"odinger operator, equivalently, to a Jacobi operator with a
finite-range perturbation. After identifying odd sequences on
\(\mathbb{Z}\) with sequences on the half-line, the interaction becomes a
boundary perturbation at the first site.

The spectral theory of Jacobi operators is closely connected with
orthogonal-polynomial methods and provides effective tools for studying
discrete eigenvalues, resonances, and threshold asymptotics; see
\cite{Simon2005a,Simon2005b,Teschl1999}. A relative-coordinate approach to
two-particle lattice systems with finite-range interactions was developed
in \cite{Valiente2010}. Threshold virtual states for continuous
Schr\"odinger operators were studied in \cite{Yafaev1979}, while related
threshold phenomena for Jacobi operators with finite-rank perturbations
were recently investigated in \cite{LMakarov:2026}.

For
\(
k\in\mathbb{T},
\)
the essential spectrum of \(H_{\lambda}(k)\) is
\[
\sigma_{\mathrm{ess}}\bigl(H_{\lambda}(k)\bigr)
=
\left[
\varepsilon_{\min}(k),\varepsilon_{\max}(k)
\right],
\]
where
\[
\varepsilon_{\min}(k)
=
4-4\cos\frac{k}{2},
\qquad
\varepsilon_{\max}(k)
=
4+4\cos\frac{k}{2}.
\]
The spectral transition is determined by the interplay between the
interaction strength \(\lambda\) and the effective hopping amplitude
\[
2\cos\frac{k}{2}.
\]
Our main analysis concerns both spectral edges. For  \(k\in(-\pi,\pi),\) we prove that \(H_{\lambda}(k)\) has a unique simple eigenvalue outside
its essential spectrum if and only if
\(|\lambda|>2\cos\frac{k}{2}.\)
Whenever this eigenvalue exists, it is given explicitly by
\[
E(k,\lambda)
=
4+\lambda
+
\frac{4\cos^{2}\frac{k}{2}}{\lambda},
\]
and the corresponding eigenfunction is obtained in explicit form in the
relative coordinate. If
\(|\lambda|<
2\cos\frac{k}{2},\)
then \(H_{\lambda}(k)\) has no eigenvalue outside \(\sigma_{\mathrm{ess}}\bigl(H_{\lambda}(k)\bigr).\)

The present analysis is not limited to deriving explicit formulas
for the eigenvalue and eigenfunction. It also completely characterizes
the critical threshold behavior of the fiber Hamiltonian for every $k\in(-\pi,\pi).$
At the critical couplings
\(
|\lambda|
=
2\cos\frac{k}{2},
\)
the discrete eigenvalue merges with the corresponding edge of the essential
spectrum, namely,
\[
E_{\mathrm{thr}}^{\pm}(k)=4\pm4\cos\frac{k}{2}
\]
The threshold equation
\[
H_{\lambda}(k)u
=
E_{\mathrm{thr}}^{\pm}(k)u
\]
has a nonzero bounded odd solution given by
\[
u_{\mathrm{res}}(r)
=
\begin{cases}
C\,(-1)^{|r|-1}\operatorname{sgn}(r), & \lambda=2\cos\frac{k}{2},\\[1mm]
C\,\operatorname{sgn}(r), & \lambda=-2\cos\frac{k}{2}.
\end{cases}
\]
Therefore, \(
u_{\mathrm{res}}\in
\ell^{\infty,\mathrm{odd}}(\mathbb Z)
\setminus
\ell^{2,\mathrm{odd}}(\mathbb Z),
\)
and the corresponding spectral edge is a threshold resonance rather
than an eigenvalue.
Thus, the critical couplings separate the regime without a
bound state from the regimes in which a unique simple
eigenvalue exists below or above the essential spectrum.

In addition, explicit strong-coupling and threshold asymptotics are obtained for the discrete eigenvalue and the corresponding eigenfunction. In the strong-coupling regime, the discrete eigenvalue admits a first-order asymptotic expansion, and the corresponding eigenfunction becomes localized at the interaction sites. As the coupling approaches its critical value, the discrete eigenvalue admits quadratic threshold asymptotics, while the corresponding eigenfunction converges pointwise to the threshold resonant solution.

The fiber \(k=\pi\) is exceptional because
\(
\cos\frac{\pi}{2}=0,
\)
and hence the hopping term vanishes. The essential spectral band therefore
collapses to the point \(\{4\}\). For every \(\lambda\neq0\), this fiber has
the simple eigenvalue
\[
E(\pi,\lambda)
=
4+\lambda,
\]
which lies below \(4\) for \(\lambda<0\) and above \(4\) for
\(\lambda>0\).

Thus, in contrast to earlier results concerned mainly with counting and
locating discrete eigenvalues, the present analysis gives an explicit
description of the threshold transition. In particular, it not only determines
the exact quasi-momentum-dependent critical coupling, derives explicit formulas
for the eigenvalue and eigenfunction, but also proves that the critical
coupling is characterized by a threshold resonance, which marks the transition
from the absence of a bound state to the existence of a unique bound state, and establishes the corresponding strong-coupling and threshold asymptotics.

The paper is organized as follows. In Section~\ref{sec2}, we define the two-particle fermionic Hamiltonian in position space. In Section~\ref{sec3}, we derive its direct-integral decomposition and obtain the fiber Hamiltonians in the relative-coordinate representation. In Section~\ref{sec4}, we study the discrete spectrum, establish the existence, uniqueness, and absence of bound states, characterize the corresponding threshold resonances, and derive the strong-coupling and threshold asymptotics of the discrete eigenvalue and the corresponding eigenfunction.

\section{The fermion-fermion lattice Hamiltonian }\label{sec2}

\subsection{Two-particle Hamiltonian in position space}
Let $\mathbb{Z}$ be the one-dimensional lattice and let $\mathbb{Z}^2:=\mathbb{Z}\times\mathbb{Z}$ be its Cartesian square.
Let $\ell^{2}(\mathbb{Z}^2)$ be a Hilbert
space of square-summable functions defined on
$\mathbb{Z}^2$ and $\ell^{2,a}(\mathbb{Z}^2)\subset\ell
^{2}(\mathbb{Z}^2)$ be the subspace  of antisymmetric functions.

We consider the free two-particle Hamiltonian associated with the fermion--fermion system acting on $\ell^{2,a}(\mathbb{Z}^2)$,
\begin{equation}\label{two-body-free-ff}
(\widehat{h}_{0}\widehat{\psi})(x_1,x_2)
=
\sum_{|s|\le1}\widehat{\varepsilon}(s)\Bigl[\widehat{\psi}(x_1+s,x_2)+\widehat{\psi}(x_1,x_2+s)\Bigr],
\end{equation}
where the function $\widehat{\varepsilon}(s)$ is given by
\begin{equation}\label{epsilon}
\widehat{\varepsilon}(s) =
\left\{\!\!\!\!
\begin{array}{rl}
2,          & \text{if } s = 0, \\
-1, & \text{if } |s| = 1, \\
0,          & \text{if } |s| > 1.
\end{array}
\right.
\end{equation}
The Hamiltonian corresponding to a nearest-neighbour interaction of strength $\lambda\in \mathbb{R}$ is given by
\begin{equation}\label{hmyuff}
\widehat{h}_{\lambda} = \widehat{h}_0 + \lambda \widehat{V},
\end{equation}
where the interaction operator $\widehat{V}$ acts as a multiplication operator:
\begin{equation}\label{vpotff}
(\widehat{V}\widehat{\psi})(x_1,x_2)
= v(x_1-x_2)\,\widehat{\psi}(x_1,x_2)
\end{equation}
with
\begin{equation}\label{vpotff1}
v(x)=
\begin{cases}
1,& |x|=1,\\
0,& \text{otherwise}.
\end{cases}
\end{equation}

\section{Fiber decomposition of the fermion-fermion Hamiltonian}\label{sec3}
In this section, we use translational invariance to reduce the fermion-fermion Hamiltonian to a direct integral of fiber operators. After passing to relative coordinates, applying the  Fourier transform, and performing a phase transformation, we obtain fiber operators acting in the odd subspace, parametrized by the total quasi-momentum $k \in \mathbb{T}.$

\subsection{Translation Invariance and Relative Coordinates}
For each $t\in\mathbb{Z}$, define the diagonal  translation operator
\[T_t:\ell^{2,a}(\mathbb{Z}^2)\to\ell^{2,a}(\mathbb{Z}^2)\]
by
\begin{equation}\label{translation}
(T_t\widehat{\psi})(x_1,x_2)=\widehat{\psi}(x_1+t,x_2+t),  \quad (x_1,x_2)\in\mathbb{Z}^2.
\end{equation}
\begin{lemma}\label{unitaryT}
The family $\{T_t\}_{t \in \mathbb{Z}}$ is a unitary representation of  $\mathbb{Z}$ on $\ell^{2,a}(\mathbb{Z}^2)$, and commutes with $\widehat {h}_{\lambda},$ that is, \[ T_t \widehat {h}_{\lambda} = \widehat{h}_{\lambda} T_t, \quad t \in \mathbb{Z}. \]
\end{lemma}
\begin{proof}
The identities
\[
T_tT_s=T_{t+s},
\qquad
T_0=I,
\qquad
T_t^{-1}=T_{-t},
\]
follow directly from the definition. A change of summation variables shows
that
\[
\|T_t\widehat\psi\|=\|\widehat\psi\|.
\]
Since \(T_t^{-1}=T_{-t}\), the isometry \(T_t\) is onto and therefore unitary. Diagonal translations preserve
antisymmetry because the same shift is applied to both particle
coordinates.

The free operator \(\widehat h_0\) consists of lattice shifts in the particle coordinates, which commute with diagonal translations. Hence
\(
T_t\widehat h_0=\widehat h_0T_t.
\)
The interaction is multiplication by
\(v(x_1-x_2)\), and
\[
(x_1+t)-(x_2+t)=x_1-x_2.
\]
Hence \(T_t\widehat V=\widehat VT_t\), which proves
\[
T_t\widehat h_\lambda
=
\widehat h_\lambda T_t.
\]
\end{proof}

We introduce new lattice coordinates
\[
R := x_2, \qquad r := x_1 - x_2 ,
\] where \(R\)
is the position of the second particle on the lattice and
\(r\) is a \textbf{\textit{relative coordinate}} (the difference between the particle coordinates).
This change of variables defines a unitary operator
\(U_0:\ell^{2}(\mathbb{Z}^2)\to\ell^{2}(\mathbb{Z}^2)\) by
\begin{equation}\label{unol}
(U_0\widehat{\psi})(R,r)=\widehat{\psi}(R+r,R).
\end{equation}
Its inverse operator is given by
\begin{equation}\label{unolinverse}
(U_0^{-1}\widehat{\varphi})(x_1,x_2)=\widehat{\varphi}(x_2, x_1-x_2).
\end{equation}
In these coordinates, diagonal translations act only on \(R\):
\begin{equation}
\label{unoltranslation}
(U_0T_tU_0^{-1}\widehat\varphi)(R,r)
=
\widehat\varphi(R+t,r).
\end{equation}
Indeed, the diagonal shift changes \(R\) to \(R+t\), whereas the
difference \(r=x_1-x_2\) remains unchanged.

\subsection{Fourier Transform and the Odd Subspace}

Let
\[
\mathcal F_{R}:
\ell^{2}(\mathbb{Z}^{2})
\longrightarrow
L^{2}\bigl(\mathbb{T};\ell^{2}(\mathbb{Z})\bigr), \quad \mathbb{T}=(-\pi,\pi]
\]
be the partial Fourier transform with respect to the variable \(R\), defined by
\begin{equation}\label{partialfourierop}
(\mathcal{F}_{R}\widehat{\varphi})(k,r)=\frac{1}{\sqrt{2\pi}}\sum_{R\in\mathbb{Z}}e^{-i kR}\widehat{\varphi}(R,r),
\qquad k\in\mathbb{T},
\end{equation}
where
$k\in\mathbb{T}$ denotes    {\it  the total quasi-momentum} of the two-particle system. For the diagonal translations \eqref{unoltranslation}
the partial Fourier transform $\mathcal{F}_{R}$ transforms them into multiplication operators
\begin{equation}\label{partialfourier}
\Bigl(\mathcal{F}_{R}(U_0T_tU_0^{-1})\widehat{\varphi}\Bigr)(k,r)=e^{ikt}(\mathcal{F}_{R}\widehat{\varphi})(k,r).
\end{equation}
We define
\[
G:
L^{2}\bigl(\mathbb T;\ell^{2}(\mathbb Z)\bigr)
\longrightarrow
L^{2}\bigl(\mathbb T;\ell^{2}(\mathbb Z)\bigr)
\]
by
\begin{equation}\label{Gaugoperator}
(Gw)(k,r)=e^{-ikr/2}w(k,r).
\end{equation}
Equivalently,
\[
G=\int_{\mathbb T}^{\oplus}G_{k}\,\,\mathrm{d}k,
\]
where $G_{k}$ is a {\it phase transformation} operator on $\ell^{2}(\mathbb{Z})$  and defined by
$$(G_{k}u)(r)=e^{-ikr/2}u(r).$$
\begin{lemma}\label{unitary-odd}
Let \(U:=G\mathcal F_{R}U_{0}.\)
Then the operator \(U\) maps the antisymmetric subspace onto the space of odd functions
in the relative variable \(r\), that is
\[
U\ell^{2,\mathrm a}(\mathbb Z^{2})
=
L^{2}\bigl(\mathbb T;\ell^{2,\mathrm{odd}}(\mathbb Z)\bigr).
\]
Moreover, the restriction
\[
U:
\ell^{2,\mathrm a}(\mathbb Z^{2})
\longrightarrow
L^{2}\bigl(\mathbb T;\ell^{2,\mathrm{odd}}(\mathbb Z)\bigr)
\]
is unitary.
\end{lemma}
 \begin{proof}
Since  \(U_0\), \(\mathcal F_R\), and \(G\) are unitary, the operator \(U=G\mathcal F_{R}U_{0}\) is unitary.

Let
\[
\widehat{\psi}\in\ell^{2,\mathrm a}(\mathbb Z^{2}),
\qquad
\widehat{\varphi}=U_{0}\widehat{\psi}.
\]
The antisymmetry of \(\widehat{\psi}\) is equivalent to
\begin{equation}
\label{phiodd}
\widehat\varphi(R,r)
=
-\widehat\varphi(R+r,-r).
\end{equation}
Applying the Fourier transform in the variable \(R\) , then \eqref{phiodd} gives
\begin{equation}
\label{wfunctionodd}
w(k,r)
=
-e^{ikr}w(k,-r).
\end{equation}
Setting \(u=Gw\), we obtain
\[
u(k,-r)
=
e^{ikr/2}w(k,-r)
=
-e^{-ikr/2}w(k,r)
=
-u(k,r).
\] Thus \(u(k,\cdot)\) is odd for almost every \(k.\)
Therefore
\[
U\ell^{2,\mathrm a}(\mathbb Z^{2})
\subset
L^{2}\bigl(\mathbb T;\ell^{2,\mathrm{odd}}(\mathbb Z)\bigr).
\]

Conversely, starting with an odd function \(u(k,\cdot)\) and applying
\(G^{-1}\), \(\mathcal F_R^{-1}\), and \(U_0^{-1}\) in reverse order
yields an antisymmetric function in
\(\ell^{2,\mathrm a}(\mathbb Z^{2})\). Thus,
\[
U\ell^{2,\mathrm a}(\mathbb Z^{2})
=
L^{2}\bigl(\mathbb T;\ell^{2,\mathrm{odd}}(\mathbb Z)\bigr).
\]
Since \(U\) is unitary on \(\ell^{2}(\mathbb Z^{2})\), its restriction to
\(\ell^{2,\mathrm a}(\mathbb Z^{2})\) is unitary.
\end{proof}

\subsection{Direct-Integral Decomposition}
\begin{lemma}\label{directintegral}
Let \(U:=G\mathcal F_{R}U_{0}.\) The  Hamiltonian \(\widehat{h}_{\lambda}\) admits the direct-integral decomposition
\[
U\widehat{h}_{\lambda}U^{-1}
=
\int_{\mathbb T}^{\oplus} H_{\lambda}(k)\,\mathrm{d}k
\]
in the Hilbert space
\(
L^{2}\bigl(\mathbb T;\ell^{2,\mathrm{odd}}(\mathbb Z)\bigr),
\)
where
\[
H_{\lambda}(k):
\ell^{2,\mathrm{odd}}(\mathbb{Z})
\to
\ell^{2,\mathrm{odd}}(\mathbb{Z})
\]
is given by
\[\big(H_{\lambda}(k)u\big)(r)=4u(r)-2\cos\frac{k}{2}\,\big[u(r+1)+u(r-1)\big]+\lambda v(r)u(r).\]
\end{lemma}

\begin{proof}
By Lemma \ref{unitary-odd}, the operator
\(
U=G\mathcal F_R U_0
\)
is unitary from
\(
\ell^{2,\mathrm a}(\mathbb{Z}^{2})
\)
onto
\(
L^{2}\bigl(\mathbb{T};\ell^{2,\mathrm{odd}}(\mathbb{Z})\bigr).
\)
Hence it is enough to study the transformed operator
\(U\widehat h_{\lambda}U^{-1}\) in the Hilbert space
 \(
L^{2}\bigl(\mathbb{T};\ell^{2,\mathrm{odd}}(\mathbb{Z})\bigr).
\)
By Lemma \ref{unitaryT}, the Hamiltonian $\widehat{h}_{\lambda}$ commutes with the diagonal translations:
\[
T_{t}\widehat{h}_{\lambda}
=
\widehat{h}_{\lambda}T_{t},
\qquad
t\in\mathbb Z.
\]
In the \((R,r)\)-coordinates using \eqref{unoltranslation}, and  applying \(\mathcal F_{R}\), we get
\[
\mathcal F_{R}U_{0}T_{t}U_{0}^{-1}\mathcal F_{R}^{-1}
=
M_{e^{ikt}},
\]
where
\[
(M_{e^{ikt}}w)(k,r)=e^{ikt}w(k,r).
\]
Since \(G\) commutes with \(M_{e^{ikt}}\), we also have
\[
UT_{t}U^{-1}=M_{e^{ikt}}.
\]
Therefore,
\[
M_{e^{ikt}}(U\widehat{h}_{\lambda}U^{-1})
=
(U\widehat{h}_{\lambda}U^{-1})M_{e^{ikt}},
\qquad
t\in\mathbb Z.
\]
Since the family $\{M_{e^{ikt}}\}_{t\in\mathbb{Z}}$ generates the algebra of multiplication operators on
\(L^{2}\bigl(\mathbb{T};\ell^{2,\mathrm{odd}}(\mathbb{Z})\bigr),\)
the standard decomposability criterion(see \cite{RSIV}, Theorem XIII.85) implies that
\[
U\widehat{h}_{\lambda\,}U^{-1}=\int_{\mathbb{T}}^{\oplus} H_{\lambda}(k)\,\mathrm{d}k,
\]
for a measurable family of bounded operators
\(
H_{\lambda}(k)
\)
acting in
\(\ell^{2,\mathrm{odd}}(\mathbb{Z}).\)

We now compute the fiber operator \(H_{\lambda}(k)\).

Let $U_1:=\mathcal F_{R}U_{0}.$ Before applying \(G\), the operator \(U_1\widehat{h}_{\lambda}U_{1}^{-1}\) has the fiber form
\begin{equation}\label{tildah}
\big(\widetilde{H}_{\lambda}(k)w\big)(r)
=4w(r)-\big[(1+e^{-ik})w(r+1)+(1+e^{ik})w(r-1)\big]+\lambda v(r)w(r).
\end{equation}
Since
\(U=GU_{1},\) the fiber operator after the phase transformation is
\[
H_{\lambda}(k)
=
G_{k}\widetilde{H}_{\lambda}(k)G_{k}^{-1}
\big|_{\ell^{2,\mathrm{odd}}(\mathbb{Z})}.
\]
Let
\(
w=G_{k}^{-1}u.
\)
Then
\(
w(r)=e^{ikr/2}u(r).
\)
Therefore
\[\big(H_{\lambda}(k)u\big)(r)=e^{-ikr/2}\left(\widetilde{H}_{\lambda}(k)w\right)(r).\]
 Using  \eqref{tildah} we get
\[
\begin{aligned}
\big(H_{\lambda}(k)u\big)(r)
&=
e^{-ikr/2}\big(\widetilde{H}_{\lambda}(k)w\big)(r) \\
&=
4u(r)
-
(1+e^{-ik})e^{ik/2}u(r+1)
-
(1+e^{ik})e^{-ik/2}u(r-1)
+
\lambda v(r)u(r).
\end{aligned}
\]
Since
\[
(1+e^{-ik})e^{ik/2}
=
(1+e^{ik})e^{-ik/2}
=
2\cos\frac{k}{2},
\]
we obtain
\[
(H_{\lambda}(k)u)(r)
=
4u(r)
-
2\cos\frac{k}{2}\Big[u(r+1)+u(r-1)\Big]
+
\lambda v(r)u(r).
\]
Since \(v\) is even, this operator preserves the odd subspace. Hence
\[
H_{\lambda}(k):
\ell^{2,\mathrm{odd}}(\mathbb{Z})
\longrightarrow
\ell^{2,\mathrm{odd}}(\mathbb{Z})
\]
is well-defined, and
\[
U\widehat{h}_{\lambda}U^{-1}
=
\int_{\mathbb T}^{\oplus}H_{\lambda}(k)\,\mathrm{d}k.
\]
\end{proof}

\section{Main spectral result for the fermion-fermion fiber operator}\label{sec4}
In this section, we study the spectral properties of the fermion-fermion
fiber operators \(H_\lambda(k)\), \(k\in\mathbb T\). We first describe their
essential spectrum and then study isolated eigenvalues below it.

By the direct-integral representation obtained in Lemma \ref{directintegral}, the spectral analysis of
\(\widehat h_{\lambda}\) is reduced to the study of the fiber operators
\[
H_{\lambda}(k):
\ell^{2,\mathrm{odd}}(\mathbb{Z})
\to
\ell^{2,\mathrm{odd}}(\mathbb{Z}),\qquad k\in\mathbb T,
\]
defined by
\[
H_{\lambda}(k)=H_{0}(k)+\lambda V,
\]
where
\begin{align*}
(H_0(k)u)(r)&=4u(r)-2\cos\frac{k}{2}\,[u(r+1)+u(r-1)],
\quad u\in \ell^{2,\mathrm{odd}}(\mathbb{Z}),
\intertext{and}
(Vu)(r)&=v(r)u(r),
\quad u\in \ell^{2,\mathrm{odd}}(\mathbb{Z}).
\end{align*}
Here the function  \(v(\cdot)\) is defined in \eqref{vpotff1}.

\subsection{Essential spectrum of the fiber operator \(H_{\lambda}(k)\)}
Let
\[
L^{2,\mathrm{odd}}(\mathbb{T}) = \{ f \in L^{2}(\mathbb{T}) : f(-p) = -f(p) \}.
\]
We describe the essential spectrum of the fiber operator \( H_{\lambda}(k) \).
\begin{lemma}\label{essentialspectr}
For every fixed \( k\in\mathbb{T} \), the essential spectrum of the operator \(H_{\lambda}(k)\)
is given by
\[
\sigma_{\mathrm{ess}}(H_{\lambda}(k))=[\varepsilon_{\min}(k), \varepsilon_{\max}(k)].
\]
where
\[\varepsilon_{\min}(k)=4-4\cos \frac{k}{2} \quad\text{and}\quad   \varepsilon_{\max}(k)=4+4\cos \frac{k}{2}.\]
\end{lemma}

\begin{proof}
Applying the Fourier transform
\[
(\mathcal Fu)(p)=\frac{1}{\sqrt{2\pi}}\sum_{r\in\mathbb{Z}}u(r)e^{-ipr},\qquad p\in\mathbb{T},
\]
for the operator \( H_{0}(k)\),
we obtain
\[
[\mathcal F H_{0}(k)\mathcal F^{-1}f](p)=\varepsilon_k(p)f(p),
\]
where
\[
\varepsilon_k(p)=4-4\cos \frac{k}{2}\cos p.
\]
Hence, \(H_{0}(k)\) is unitarily equivalent to the multiplication operator by \(\varepsilon_k(p) \) on \( L^{2,\mathrm{odd}}(\mathbb{T}) \). Since
\[\min\limits_{p\in\mathbb{T}}\varepsilon_k(p)=\varepsilon_{\min}(k), \quad  \max\limits_{p\in\mathbb{T}}\varepsilon_k(p)=\varepsilon_{\max}(k),\]
we conclude that
\[
\sigma(H_{0}(k))=\sigma_{\mathrm{ess}}(H_{0}(k))=[\varepsilon_{\min}(k),\varepsilon_{\max}(k)]=\bigl[4-4\cos \frac{k}{2},\;4+4\cos \frac{k}{2}\bigr].
\]
Since $v(\cdot)$ has finite support, the operator $V$ is finite rank
and Weyl's theorem on invariance of the essential spectrum under compact perturbations yields
\[
\sigma_{\mathrm{ess}}(H_{\lambda}(k))=\sigma_{\mathrm{ess}}(H_{0}(k))=[\varepsilon_{\min}(k), \varepsilon_{\max}(k)].
\]
\end{proof}

\subsection{Discrete spectrum outside the essential spectrum}
In this section, we study the discrete spectrum of the fiber operators
\(H_\lambda(k)\) outside the essential spectrum. We obtain sharp conditions
for the existence and absence of isolated eigenvalues.

We now state the  result on the discrete spectrum outside the essential
spectrum.

\begin{theorem}\label{maintheorem}
The following assertions hold:
\begin{enumerate}
\item[(i)] Let $k\in(-\pi,\pi)$ and $|\lambda|>2\cos\frac{k}{2}$. Then the operator
$H_{\lambda}(k)$ has a unique simple eigenvalue
outside the essential spectrum. This eigenvalue is given by
\begin{equation}\label{eigenformula}
E(k,\lambda)
=4+\lambda+\frac{4\cos^2\frac{k}{2}}{\lambda}.
\end{equation}
This eigenvalue lies below the essential spectrum when
\(
\lambda<-2\cos\frac{k}{2},
\)
and above the essential spectrum when
\(
\lambda>2\cos\frac{k}{2}.
\)
Moreover, the function $\lambda\mapsto E(k,\lambda)$ is
strictly increasing on \(\left(-\infty,-2\cos\frac{k}{2}\right)\) and
\(\left(2\cos\frac{k}{2},+\infty\right).\)
A corresponding eigenfunction has the form
\begin{equation}\label{efunctionff}
u(r)=C\,\alpha^{|r|-1}\operatorname{sgn}(r), \qquad C\neq0 \qquad r\in\mathbb{Z},
\end{equation}
where
\begin{equation}\label{alpharelation}
\alpha=-\frac{2\cos\frac{k}{2}}{\lambda}, \qquad 0<|\alpha|<1.
\end{equation}

\item[(ii)]Let $k=\pi$ and $\lambda\neq0$. Then the operator
$H_{\lambda}(\pi)$ has a unique simple eigenvalue
outside the essential spectrum. This eigenvalue is given by
\begin{equation}\label{eigenff-pi}
E(\pi,\lambda)
=4+\lambda.
\end{equation}
The corresponding eigenfunction has the form
\[
u(r)=
\begin{cases}
C, & r=1,\\[0.4ex]
-C, & r=-1,\\[0.4ex]
0, & |r|\neq 1,
\end{cases}
\qquad C\neq 0.
\]

\item[(iii)] Let $k\in(-\pi,\pi)$ and  $|\lambda|<2\cos\frac{k}{2}.$ Then the operator
$H_{\lambda}(k)$ has no eigenvalues outside the essential spectrum.
\end{enumerate}
\end{theorem}

\begin{proof}
Let $E\in\mathbb{R}$ and $u\neq0, u\in\ell^{2,\mathrm{odd}}(\mathbb{Z})$ satisfy the eigenvalue equation
\begin{equation}\label{eigenequation}
H_{\lambda}(k)u=Eu.
\end{equation}
The eigenvalue equation outside the support of the interaction has the form
\[
E u(r) = 4 u(r) - 2\cos\frac{k}{2} \left[ u(r+1) + u(r-1) \right], \quad |r| \geq 2.
\]
At the interaction sites $r = \pm 1$, it is given by
\[
E u(1) = 4 u(1) - 2\cos\frac{k}{2} \left[ u(2) + u(0) \right] + \lambda u(1).
\]
Since $u\in\ell^{2,\mathrm{odd}}(\mathbb Z)$, we have $u(-r)=-u(r)$ for all $r\in\mathbb Z$. In particular, this implies that $u(0)=0$. Therefore the eigenvalue problem is reduced to the following boundary value problem on the half-line $r\ge 1$:
\begin{equation}\label{eigensystem1}
\begin{cases}
E\,u(r)=4 u(r) - 2\cos\frac{k}{2}\, \bigl[ u(r+1) + u(r-1) \bigr], & r\ge 2,\\
E\,u(1)= 4 u(1) - 2\cos\frac{k}{2}\, u(2) + \lambda u(1),& r=1,
\end{cases}
\end{equation}
with \(u(0)=0,\,\, u\in\ell^2(\mathbb{N}),\) where $\mathbb{N}:=\{1,2,3,...\}$.

{\bf (i)} Let $k\in(-\pi,\pi).$ We introduce
\[
a(k) := 2\cos\frac{k}{2}.
\]
Since $k\in(-\pi,\pi)$, we have \(a(k)>0.\) Moreover,
\[
\sigma_{\mathrm{ess}}\bigl(H_{\lambda}(k)\bigr) = [4-2a(k),\, 4+2a(k)],
\]
and  an eigenvalue outside the essential spectrum must satisfy
\[
\left|\frac{4 - E}{a(k)}\right|>2.
\]
 For $r \geq 2$, we solve the second-order difference equation
\[
E u(r) = 4u(r) - a(k)\bigl(u(r+1) + u(r-1)\bigr).
\]
Looking for solutions in the form
\[
u(r) = \beta^r,
\]
we obtain
\begin{equation}\label{chareq}
E = 4 - a(k)\bigl(\beta + \beta^{-1}\bigr).
\end{equation}
This implies the following quadratic equation:
\begin{equation}\label{quadraticeq}
\beta^2-\frac{4-E}{a(k)}\beta+1=0.
\end{equation}
Since
\[
\left|\frac{4 - E}{a(k)}\right|>2,\]
 the characteristic equation \eqref{quadraticeq} has two distinct real reciprocal roots. Let $\alpha$ denote the unique root satisfying
$$0<|\alpha|<1.$$ Hence
\[
\beta_{1} =\alpha^{-1}, \qquad \beta_{2} = \alpha,
\]
where \[\alpha:=\dfrac{\frac{4 - E}{a(k)}-\sqrt{\left[\frac{4 - E}{a(k)}\right]^2-4}}{2},\quad 0 < |\alpha| < 1.\]
Consequently, every solution of  the second-order difference equation on the half-line is of the form
\[
u(r) = C_1 \alpha^r + C_2 \alpha^{-r}, \qquad r \geq 1.
\]
Since $0<|\alpha|<1,$ we have
\[
|\alpha^{-r}| \to \infty \quad \text{as } r \to \infty.
\]
Therefore the condition \(u \in \ell^2(\mathbb{N})\) implies
\(C_2 = 0.\) Thus every eigenfunction corresponding to an eigenvalue outside the essential spectrum is necessarily of the form
\begin{equation}\label{gensolution}
u(r)=C_1\alpha^r,
\quad r\geq 1,
\quad C_1\neq 0,
\quad 0<|\alpha|<1.
\end{equation}
Equivalently, after changing the nonzero constant, this solution can be written in the shifted form
\begin{equation}\label{gensolution2}
u(r)=C\alpha^{r-1},
\quad r\geq 1,
\quad C\neq 0,
\quad 0<|\alpha|<1.
\end{equation}
The shifted form is more convenient because
\[u(1)=C,\qquad u(2)=C\alpha\]
Substituting $u(r)=C\alpha^{r-1}$ into the second equation of \eqref{eigensystem1} and since \( C \neq 0 \), we obtain
\[
E = 4 - a(k)\alpha + \lambda.
\]
On the other hand, from the characteristic relation \eqref{chareq},
\begin{equation}\label{eigenvalue}
E=4-a(k)\left(\alpha+\alpha^{-1}\right).
\end{equation}
Comparing these two expressions for $E$, we get
\[
4-a(k)\alpha+\lambda = 4-a(k)\alpha-a(k)\alpha^{-1}.
\]
Hence
\begin{equation*}
\lambda=-a(k)\alpha^{-1},
\end{equation*}
and therefore
\[
\alpha=-\frac{a(k)}{\lambda} = -\frac{2\cos\frac{k}{2}}{\lambda}.
\]
Since $0<|\alpha|<1$, this is equivalent to
\[
|\lambda|>a(k) = 2\cos\frac{k}{2}.
\]
Thus an eigenvalue outside the essential spectrum can exist only if
\[
|\lambda|>2\cos\frac{k}{2}.
\]
Substituting the value of $\alpha$  into  \eqref{eigenvalue}, we obtain the corresponding eigenvalue
\begin{equation}\label{eq:Eff}
E(k,\lambda):=E
= 4 - 2\alpha\cos\tfrac{k}{2} + \lambda
= 4 + \lambda + \frac{4\cos^2\frac{k}{2}}{\lambda}.
\end{equation}
Since $0<|\alpha|<1$, the function
\[
u(r)=C\alpha^{|r|-1} \operatorname{sgn}(r), \,\, r\in\mathbb Z,\ C\neq 0,
\]
belongs to $\ell^{2,\mathrm{odd}}(\mathbb Z)$.
Moreover, by construction, it satisfies \eqref{eigensystem1} with
\(
E(k,\lambda).
\)
Hence $u$ is a corresponding eigenfunction.

We now show that $E(k,\lambda)$ lies outside the essential spectrum.
Using
\[
\alpha = -\frac{a(k)}{\lambda}, \qquad 0 < |\alpha| < 1,
\]
we have
\[
\alpha+\alpha^{-1}>2
\quad\text{for }\quad \lambda<0,
\]
and
\[
\alpha+\alpha^{-1}<-2
\quad\text{for }\quad\lambda>0.
\]
On the other hand, from \eqref{eigenvalue},
\[
E(k,\lambda) = 4 - a(k)(\alpha + \alpha^{-1}).
\]
Therefore
\[
E(k,\lambda)<4-2a(k)=\varepsilon_{\min}(k)=\inf\sigma_{\mathrm{ess}}\bigl(H_{\lambda}(k)\bigr), \quad\text{for }\lambda<0,
\]
and
\[
E(k,\lambda)
>
4+2a(k)
=\varepsilon_{\max}(k)=
\sup\sigma_{\mathrm{ess}}\bigl(H_{\lambda}(k)\bigr), \quad\text{for } \lambda>0.
\]

The above argument also proves uniqueness. Indeed, any eigenvalue outside the essential spectrum gives an $\ell^2$-solution of the form
\[
u(r) = C \alpha^{r-1}, \qquad r \geq 1, \quad 0 < |\alpha| < 1.
\]
The second equation of \eqref{eigensystem1} uniquely determines
\[
\alpha = -\frac{2\cos\frac{k}{2}}{\lambda}.
\]
Hence $E(k,\lambda)$ is also uniquely determined by
\[
E(k,\lambda) = 4 - a(k)(\alpha + \alpha^{-1}).
\]
Therefore $H_{\lambda}(k)$ has exactly one eigenvalue outside the essential spectrum.

Furthermore, every eigenfunction corresponding to this eigenvalue is a scalar multiple of
\[
\alpha^{|r|-1} \operatorname{sgn}(r).
\]
Hence
\[
\dim \ker \bigl( H_{\lambda}(k) - E(k,\lambda) \bigr) = 1.
\]
Since $H_{\lambda}(k)$ is self-adjoint, this isolated eigenvalue is simple.

For \(|\lambda|>2\cos\frac{k}{2},\) one has
\[
\partial_{\lambda}E(k,\lambda)=1-\frac{4\cos^2\frac{k}{2}}{\lambda^2}>0.
\]
Hence the function \( \lambda\mapsto E(k,\lambda)\) is strictly increasing on \((-\infty,-2\cos\frac{k}{2})\) and \((2\cos\frac{k}{2},+\infty)\)

{\bf{(ii)}} Let $k=\pi$. Then
\[
a(k)=2\cos\frac{\pi}{2}= 0,
\]
and the system \eqref{eigensystem1} becomes
\begin{equation}\label{eigensystem2}
\begin{cases}
E\,u(r)=4 u(r) & r\ge 2,\\
E\,u(1)= (4+\lambda) u(1),
\end{cases}
\end{equation}
By Lemma \ref{essentialspectr},
\[
\sigma_{\mathrm{ess}}(H_{\lambda}(\pi))=\{4\}.
\]
Let $u \neq 0$  be an eigenfunction corresponding to an eigenvalue $E\notin\sigma_{\mathrm{ess}}(H_{\lambda}(\pi)).$
Since $E\neq4,$  the first equation of the system \eqref{eigensystem2} yields
\[
(E-4)u(r)=0,\quad |r|\geq 2,
\]
and therefore
\(u(r)=0,\,\,\, |r|\geq 2.\) If \(u(1)=0\), then  \[u(-1)=-u(1)=0,\] and hence \(u(r)=0\) contradicting the assumption that \(u\neq0\). Therefore \(u(1)\neq0\). The second equation of \eqref{eigensystem2}  gives
\[
E(\pi,\lambda):=E=4+\lambda.
\]
Since $\lambda\neq0$,
$E\notin\sigma_{\mathrm{ess}}(H_{\lambda}(\pi)).$
The corresponding eigenfunction is
\[
u(r)=
\begin{cases}
C, & r=1,\\[0.4ex]
-C, & r=-1,\\[0.4ex]
0, & |r|\neq 1,
\end{cases}
\qquad C\neq 0.
\]
The corresponding eigenspace is one-dimensional. Hence the eigenvalue is simple.

{\bf (iii)} Assume
\(
|\lambda| < 2\cos \frac{k}{2}.
\)

If \(\lambda = 0\), then
\[
H_\lambda(k) = H_0(k),
\]
and therefore
\[
\sigma(H_0(k)) = \sigma_{\mathrm{ess}}(H_0(k)),
\]
so no eigenvalue exists outside the essential spectrum.

Now let
\(
0 < |\lambda| < 2\cos \frac{k}{2}.
\)
We show that in this case the operator has no eigenvalues below the essential spectrum.

Suppose, to the contrary, that such an eigenvalue exists. Then, by the argument used in the proof of assertion (i), the corresponding solution must be of the form
$$
u(r)=C\alpha^{|r|-1}\operatorname{sgn}(r),\qquad C\neq 0,
$$
where
$0<|\alpha|<1.$
Moreover, the second equation of \eqref{eigensystem1} gives
\[
\alpha = -\frac{2\cos\frac{k}{2}}{\lambda}.
\]
Since
$
0<|\lambda|<2\cos\frac{k}{2},
$
we obtain
$$
|\alpha|=\frac{2\cos\frac{k}{2}}{|\lambda|}>1.
$$
This contradicts the necessary condition $0<|\alpha|<1$. Hence no eigenvalue outside the essential spectrum exists for
$|\lambda|<2\cos\frac{k}{2}$.
\end{proof}

\begin{corollary}[\bf Strong-coupling asymptotics]
  Let \(k\in(-\pi,\pi).\)
As \(|\lambda|\to\infty\),
\[
E(k,\lambda)
=
\lambda+4+O\bigl(|\lambda|^{-1}\bigr),
\]
and
\[
\alpha
=
-\frac{2\cos\frac{k}{2}}{\lambda}
\longrightarrow0.
\]
Consequently, the normalized eigenfunction converges in
\(\ell^{2,\mathrm{odd}}(\mathbb Z)\) to
\[
\frac{1}{\sqrt{2}}
\bigl(\delta_{1}-\delta_{-1}\bigr).
\]
Thus, in the strong-coupling limit, the bound state becomes localized at the interaction sites \(r=\pm 1\).
\end{corollary}
\begin{proof}
The assertions follow immediately from the explicit formulas for $E(k,\lambda),$ $\alpha,$ and the eigenfunction given in Theorem \ref{maintheorem}.
\end{proof}

\begin{remark}
The exceptional case \(k=\pi\) is obtained as the limit of the formulas derived for
\(k\in(-\pi,\pi)\). Let \(\lambda<0\) be fixed. Then, for
\(k\in(-\pi,\pi)\) sufficiently close to \(\pi\), the condition
\[
|\lambda|>2\cos\frac{k}{2}
\]
is satisfied, and
\[
\alpha(k)=-\frac{2\cos\frac{k}{2}}{\lambda}\to0
\qquad \text{as } \quad k\to\pi .
\]
Consequently,
\[
E(k,\lambda)
=
4+\lambda+\frac{4\cos^2{\frac{k}{2}}}{\lambda}
\to
4+\lambda
=
E(\pi,\lambda).
\]
The corresponding normalized eigenfunction is
\[
u_k(r)
=
\left(\frac{1-\alpha(k)^2}{2}\right)^{1/2}
\alpha(k)^{|r|-1}\operatorname{sgn}(r),
\qquad r\in\mathbb{Z}.
\]
Then
\[
u_k(r)\to \frac1{\sqrt2}(\delta_1-\delta_{-1})
\qquad \text{in } \,\, \ell^{2,\mathrm{odd}}(\mathbb{Z}).
\]
\end{remark}

\subsection{Threshold resonance.}
Let
\[\ell^{\infty}(\mathbb{Z})=\{u:\mathbb{Z}\to\mathbb{C}:\sup_{r\in\mathbb{Z}}|u(r)|<\infty\}\]
be the Banach space of bounded complex-valued functions on \(\mathbb{Z}\).
We also set
\[\ell^{\infty,\mathrm{odd}}(\mathbb{Z})=\{u\in\ell^{\infty}(\mathbb{Z}):u(-r)=-u(r),\ r\in\mathbb{Z}\}.
\]
Thus \(\ell^{\infty,\mathrm{odd}}(\mathbb{Z})\) is the subspace of bounded odd
functions on \(\mathbb Z\).

We introduce the notation
\[
E_{\mathrm{thr}}^{-}(k):=\varepsilon_{\min}(k),
\qquad
E_{\mathrm{thr}}^{+}(k):=\varepsilon_{\max}(k),
\]
for the lower and upper thresholds of the essential spectrum.
\begin{definition}
Let \(k\in(-\pi,\pi)\) and \(\lambda\neq0\). We say that
\(H_{\lambda}(k)\) has a resonance at a threshold
\(E_{\mathrm{thr}}^{\pm}(k)\)
if the equation
\[
H_{\lambda}(k)u=E_{\mathrm{thr}}^{\pm}(k)u
\]
has a nonzero solution
\[u\in\ell^{\infty,\mathrm{odd}}(\mathbb{Z})
\setminus
\ell^{2,\mathrm{odd}}(\mathbb{Z}).
\]
Such a solution is called a resonant solution.
\end{definition}
The following theorem characterizes the threshold resonances of $H_\lambda(k).$
\begin{theorem}
Let \(k\in(-\pi,\pi)\) and
\(\lambda=\pm2\cos\frac{k}{2}.\)
Then \(H_{\lambda}(k)\) has a resonance at the threshold
\(E_{\mathrm{thr}}^{\pm}(k).\)
A corresponding resonant solution is given by
\begin{equation}\label{resfunctionff}u_{\mathrm{res}}(r)=\begin{cases}
C\,(-1)^{|r|-1}\operatorname{sgn}(r), & \lambda=2\cos\frac{k}{2}\\
C\,\operatorname{sgn}(r), & \lambda=-2\cos\frac{k}{2},
\end{cases}
\qquad C\neq 0,
\end{equation}
and belongs to
\(
\ell^{\infty,\mathrm{odd}}(\mathbb Z)
\setminus
\ell^{2,\mathrm{odd}}(\mathbb Z).
\)
\end{theorem}
\begin{proof} We prove the theorem only for
\(\lambda=-2\cos\frac{k}{2},\) since the case \(\lambda=2\cos\frac{k}{2}\)
is identical, with \(\alpha=-1\) replacing \(\alpha=1\).

Let \(k\in(-\pi,\pi)\) and \(\lambda=-2\cos\frac{k}{2}=-a(k).\)
By the Theorem \ref{maintheorem}, the relation
\[
\lambda=-a(k)\alpha^{-1},
\]
implies that
\(
\alpha=1.
\)
Therefore,
by \eqref{eigenformula},
\[
E\left(k,-2\cos\frac{k}{2}\right)=4-4\cos\frac{k}{2}=E_{\mathrm{thr}}^{-}(k).
\]
Moreover, it follows \eqref{efunctionff} that
\[
u(r)=C\operatorname{sgn}(r), \quad C\neq0.
\]
Since
\(
u\in\ell^{\infty,\mathrm{odd}}(\mathbb Z)\setminus\ell^{2,\mathrm{odd}}(\mathbb Z),\)
and satisfies the equation
\[
H_\lambda(k)u=E_{\mathrm{thr}}^{-}(k)u,
\]
the operator \(H_\lambda(k)\) has a resonance at the lower threshold of its
essential spectrum.
\end{proof}
\begin{corollary}[\bf Threshold asymptotics]
Let \(k\in(-\pi,\pi)\). As
\(
\lambda\rightarrow\pm2\cos\frac{k}{2},
\)
from the eigenvalue regime, the discrete eigenvalue satisfies
\[E(k,\lambda)\longrightarrow E_{\mathrm{thr}}^{\pm}(k),\]
and
\[
\bigl|E(k,\lambda)-E_{\mathrm{thr}}^{\pm}(k)\bigr|
=
\frac{\bigl(|\lambda|-2\cos\frac{k}{2}\bigr)^2}{|\lambda|}
=
O\!\left(
\bigl(|\lambda|-2\cos\tfrac{k}{2}\bigr)^2
\right).
\]
Furthermore, choosing \(C=1\) in the representations of the eigenfunction \eqref{efunctionff}
and the resonant solution \eqref{resfunctionff},
\[
u(r)\longrightarrow u_{\mathrm{res}}(r),
\qquad r\in\mathbb Z,
\]
pointwise on $\mathbb Z.$
\end{corollary}
\begin{proof}
The asymptotic formula follows from \eqref{eigenformula}. Moreover, choosing $C=1$ in \eqref{efunctionff} and \eqref{resfunctionff}, the pointwise convergence $u(r)\rightarrow u_{\mathrm{res}}(r)$ follows directly from \eqref{alpharelation}.
\end{proof}

\end{document}